\documentclass[12pt,fleqn]{article}
\usepackage{amsmath}
\usepackage{amsthm}
\usepackage{amsfonts}
\usepackage{amssymb}
\usepackage{fullpage}
\usepackage{color}
\usepackage{times}

\newtheorem{prop}{Proposition}
\newtheorem{lemma}{Lemma}
\newtheorem{corollary}{Corollary}
\begin{document}
\renewcommand{\thefootnote}{\fnsymbol{footnote}}
\title{On Backus average for generally anisotropic layers}
\author{
Len Bos\footnote{
Dipartimento di Informatica, Universit\`a di
Verona, {\tt leonardpeter.bos@univr.it}}\,,
David R. Dalton\footnote{
Department of Earth Sciences, Memorial University of Newfoundland,
{\tt dalton.nfld@gmail.com}}\,, 
Michael A. Slawinski\footnote{
Department of Earth Sciences, Memorial University of Newfoundland,
{\tt mslawins@mac.com}}\,,
Theodore Stanoev\footnote{
Department of Earth Sciences, Memorial University of Newfoundland,
{\tt theodore.stanoev@gmail.com}}
}
\date{June 22, 2016}
\maketitle
\renewcommand{\thefootnote}{\arabic{footnote}}
\setcounter{footnote}{0}
\section*{Abstract}
In this paper, following the Backus (1962) approach, we examine expressions for elasticity parameters of a homogeneous
generally anisotropic medium that is long-wave-equivalent to a stack of
thin generally  anisotropic  layers.
These expressions reduce to the results of Backus (1962) for the case of isotropic and transversely isotropic layers.

In the over half-a-century since the publications of Backus (1962) there have been numerous publications applying and extending that formulation.
However, neither George Backus nor the authors of the present paper are aware of further examinations of the mathematical underpinnings of the original formulation; hence this paper.

We prove that---within the long-wave approximation---if the thin layers obey stability conditions then so does 
the equivalent medium. 
We examine---within the Backus-average context---the approximation of the average of a product as the product of averages, which underlies the averaging process.

In the presented examination we use the expression of Hooke's law as a tensor equation; in other words, we use Kelvin's---as opposed to Voigt's---notation.
In general, the tensorial notation allows us to conveniently examine effects due to rotations of coordinate systems.
%%%%%%%%%%%%%%%%%%%%%%%%%%%%
\section{Introduction and historical background}
%%%%%%%%%%%%%%%%%%%%%%%%%%%%
The study of properties of materials as a function of scale has occupied researchers for decades.
Notably, the discipline of continuum mechanics originates, at least partially,  from such a consideration.
Herein, we focus our attention on the effect of a series of thin and laterally homogeneous layers on a long-wavelength wave.
These layers are composed of generally anisotropic Hookean solids.

Such a mathematical formulation serves as a quantitative analogy for phenomena examined in seismology.
The effect of seismic disturbances---whose wavelength is much greater than the size of encountered inhomogeneities---is tantamount to the smearing of the mechanical properties of such inhomogeneities.
The mathematical analogy of this smearing is expressed as averaging.
The result of this averaging is a homogeneous anisotropic medium to which we refer as an {\sl equivalent medium}.

We refer to the process of averaging as {\sl Backus averaging}, which is a common nomenclature in seismology.
However, several other researchers have contributed to the development of this method.

Backus (1962) built on the work of Rudzki (1911), Riznichenko (1949), Thomson (1950), Haskell (1953), White and Angona (1955), Postma (1955), Rytov (1956), Helbig (1958) and Anderson (1961) to show that a homogeneous transversely isotropic medium with a vertical symmetry axis could be a long-wave equivalent to a stack of thin isotropic or transversely isotropic layers.
In other words, the Backus average of thin layers appears---at the scale of a long wavelength---as a homogeneous transversely isotropic medium.

In this paper, we discuss the mathematical underpinnings of the Backus (1962) formulation.
To do so, we consider a homogeneous generally anisotropic medium that is a long-wave equivalent to a stack of thin generally anisotropic layers.
The cases discussed explicitly by Backus (1962) are special cases of this general formulation.
%%%%%%%%%%%%%%%%%%%%%%%%%%%%
\section{Averaging Method}
%%%%%%%%%%%%%%%%%%%%%%%%%%%%
\subsection{Assumptions}
\label{sub:Ass}
%%%%%%%%%%%%%%%%%%%%%%%%%%%%
We assume the lateral homogeneity of Hookean solids consisting of a series of layers that are parallel to the $x_1x_2$-plane and have an infinite lateral extent.
We subject this series to the same traction above and below, independent of time or lateral position.
It follows that the stress tensor components~$\sigma_{i3}$\,, where $i\in\{1,2,3\}$\,, are constant throughout the strained medium, due to the requirement of equality of traction across interfaces (e.g., Slawinski (2015), pp.~430--432), and to the definition of the stress tensor,
\begin{equation*}
T_i=\sum\limits_{j=1}^3\sigma_{ij}n_j\,,
\qquad i\in\{1,2,3\}
\,,
\end{equation*}
where $T$ is traction and $n$ is the unit normal to the interface.
No such equality is imposed on the other three components of this symmetric tensor; $\sigma_{11}$\,, $\sigma_{12}$ and $\sigma_{22}$ can vary wildly along the $x_3$-axis due to changes of elastic properties from layer to layer.

Furthermore, regarding the strain tensor, we invoke the kinematic boundary conditions that require no slippage or separation between layers; in other words, the corresponding components of the displacement vector,~$u_1$\,, $u_2$ and $u_3$\,, must be equal to one another across the interface (e.g., Slawinski (2015), pp.~429--430).

These conditions are satisfied  if $u$ is continuous.
Furthermore, for parallel layers, its derivatives with respect to~$x_1$ and $x_2$\,, evaluated along the $x_3$-axis, remain small.
However, its derivatives with respect to~$x_3$\,, evaluated along that axis, can vary wildly.

The reason for the differing behaviour of the derivatives resides within Hooke's law,
\begin{equation}
\label{eq:Hooke}
\sigma_{ij}=\sum\limits_{k=1}^3\sum_{\ell=1}^3c_{ijk\ell}\varepsilon_{k\ell}
\,,\qquad
i,j\in\{1,2,3\}
\,,
\end{equation}
where
\begin{equation}
\label{eq:strain}
\varepsilon_{k\ell}:=\frac{1}{2}\left(\frac{\partial u_k}{\partial x_\ell}+\frac{\partial u_\ell}{\partial x_k}\right),
\qquad 
k,\ell\in\{1,2,3\}
\,.
\end{equation}
Within each layer, derivatives are linear functions of the stress tensor.
The derivatives with respect to~$x_1$ and $x_2$ remain within a given layer; hence, the linear relation remains constant.
The derivatives with respect to~$x_3$ exhibit changes due to different properties of the layers.

In view of definition~(\ref{eq:strain}), $\varepsilon_{11}$\,, $\varepsilon_{12}$ and $\varepsilon_{22}$ vary slowly along  the $x_3$-axis.
On the other hand, $\varepsilon_{13}$\,, $\varepsilon_{23}$ and $\varepsilon_{33}$ can vary wildly along that axis.

Herein, we assume that the elasticity parameters are expressed with respect to the same coordinate system for all layers.
However, this {\it a priori\/} assumption can be readily removed by rotating, if necessary, the coordinate systems to express them in the same orientation.
%%%%%%%%%%%%%%%%%%%%%%%%%%%%
\subsection{Definitions}
%%%%%%%%%%%%%%%%%%%%%%%%%%%%
Following the definition proposed by Backus (1962), the average of the function $f(x_3)$ of ``width''~$\ell'$ is the moving average given by
\begin{equation}
\label{eq:BackusOne}
\overline f(x_3):=\int\limits_{-\infty}^\infty w(\zeta-x_3)f(\zeta)\,{\rm d}\zeta
\,,
\end{equation}
where the weight function,~$w(x_3)$\,, is an  {\sl approximate identity}, which is an approximate Dirac delta that acts like the delta centred at $x_3=0$\,, with the following properties:
\begin{equation*}
w(x_3)\geqslant0\,,\!\!
\quad w(\pm\infty)=0\,,\!\!
\quad
\int\limits_{-\infty}^\infty w(x_3)\,{\rm d}x_3=1\,,\!\!
\quad
\int\limits_{-\infty}^\infty x_3w(x_3)\,{\rm d}x_3=0\,,\!\!
\quad
\int\limits_{-\infty}^\infty x_3^2w(x_3)\,{\rm d}x_3=(\ell')^2\,.
\end{equation*}
These properties define $w(x_3)$ as a probability-density function with mean~$0$ and standard deviation~$\ell'$\,, explaining the use of the term ``width'' for $\ell'$\,.

To understand the effect of such averaging, which is tantamount to smoothing by a wave, we may consider its effect on the pure frequency, $f(x_3)=\exp(- i \omega x_3)$\,,
\begin{equation*}
\overline f(x_3)=\int\limits_{-\infty}^\infty w(\zeta-x_3)f(\zeta)\,{\rm d}\zeta
=\int\limits_{-\infty}^\infty  w(\zeta-x_3) \exp(-\iota\omega \zeta)\,{\rm d}\zeta
=\int\limits_{-\infty}^\infty w(u)\exp(-\iota\omega(u+x_3))\,{\rm d}u
\,,
\end{equation*}
where $u:=\zeta-x_3$ and $\iota:=\sqrt{-1}$\,; it follows that
\begin{equation*}
\overline f(x_3)
=\exp(-\iota\omega x_3)\int\limits_{-\infty}^\infty w(u)\exp(-\iota\omega u)\,{\rm d}u
=\exp(-\iota\omega x_3)\widehat{w}(\omega)
\,,
\end{equation*}
where $\widehat{w}(\omega)$ is the Fourier transform of $w(x_3)$\,.

If, in addition, $w(x_3)$ is an even function, then  $\widehat{w}(\omega)$ is real-valued and we may think of $\overline f(x_3)$ as the pure frequency, $\exp(-\iota\omega x_3)$\,, whose ``amplitude'' is $\widehat{w}(\omega)$\,.
The classical Riemann-Lebesgue Lemma implies that this amplitude tends to zero as the frequency goes to infinity.
To examine this decay of amplitude, we may consider a common choice for~$w(x_3)$\,, namely,  the Gaussian density,
\[w(x_3)=\frac{1}{\ell'\sqrt{2\pi}}\exp\left(-\frac{x_3^2}{2(\ell')^2}\right)\,.\]
As is well known, in this case,
\[\widehat{w}(\omega)
=\exp\left(-\frac{(\omega\,\ell')^2}{2}\right)
\,,\]
which is a multiple of the Gaussian density with standard deviation $1/\ell'$\,.
In particular, one notes the fast decay, as the product $\omega\,\ell'$ increases.

Perhaps it is useful to look at another example.
Consider
\[w(x_3)=\frac{1}{2\sqrt{3}\,\ell'}I_{[-\sqrt{3}\,\ell',\sqrt{3}\,\ell']}\,,\]
which is the uniform density on the interval $[-\sqrt{3}\,\ell',\sqrt{3}\,\ell']$\,, and which satisfies the defining properties of $w(x_3)$\,, as required.
Its Fourier transform is
\[\widehat{w}(\omega)=\frac{\sin(\sqrt{3}\omega\ell')}{\sqrt{3}\omega\ell'}\,,\]
and, as expected, this amplitude tends to zero as $\omega\rightarrow\pm\infty$\,, but at a much slower rate than in the Gaussian case; herein, the decay rate is order $1/(\omega\ell')$\,.

%%%%%%%%%%%%%%%%%%%%%%%%%%%%
\subsection{Properties}
%%%%%%%%%%%%%%%%%%%%%%%%%%%%
To perform the averaging, we use its linearity, according to which the average of a sum is the sum of the averages, $\overline{f+g}=\overline{f}+\overline{g}$\,.
Also, we use the following lemma
\begin{lemma}
\label{lem:LemDer}
The average of the derivative is the derivative of the average,
\begin{equation*}
\overline{\frac{\partial f}{\partial x_i}}=\frac{\partial}{\partial x_i}\overline f\,,\qquad
i\in\{1,2,3\}
\,.
\end{equation*}
\end{lemma}
\noindent This lemma is proved in \ref{AppLemDer}.
In \ref{AppLemStab}, we prove the lemma that ensures that the average of Hookean solids results in a Hookean solid, which is
\begin{lemma}
\label{lem:LemStab}
If the individual layers satisfy the stability condition, so does their equivalent medium.
\end{lemma}
\noindent The proof of this lemma invokes Lemma~\ref{lem:LemProd}, below.
%%%%%%%%%%%%%%%%%%%%%%%%%%%%
\subsection{Approximations}
\label{sub:Approx}
%%%%%%%%%%%%%%%%%%%%%%%%%%%%
In \ref{AppLemApp}, we state and prove a result that may be paraphrased as
\begin{lemma}
\label{lem:LemProd}
If $f(x_3)$ is nearly constant along $x_3$ and $g(x_3)$ does not vary excessively, then $\overline{fg}\approx\overline f\,\overline g$\,.
\end{lemma}

An approximation---within the physical realm---is our applying the static-case properties to examine wave propagation, which is a dynamic process.
As stated in Section~\ref{sub:Ass}, in the case of static equilibrium, $\sigma_{i3}$\,, where $i\in\{1,2,3\}$\,, are constant.
We consider that these stress-tensor components remain nearly constant along the $x_3$-axis, for the farfield and long-wavelength phenomena.
As suggested by Backus (1962), the concept of a long wavelength can be quantified as $\kappa\,\ell'\ll 1$\,, where $\kappa$ is the wave number.
Similarly, we consider that $\varepsilon_{11}$\,, $\varepsilon_{12}$ and $\varepsilon_{22}$ remain slowly varying along  that axis.

Also, we assume that waves propagate perpendicularly, or nearly so, to the interfaces. Otherwise, due to inhomogeneity between layers, the proportion of distance travelled in each layer is a function of the source-receiver offset, which---in principle---entails that averaging requires different weights for each layer depending on the offset (Dalton and Slawinski, 2016).
%%%%%%%%%%%%%%%%%%%%%%%%%%%%
\section{Equivalent-medium elasticity parameters}
%%%%%%%%%%%%%%%%%%%%%%%%%%%%
Consider the constitutive equation for a generally anisotropic Hookean solid,
\begin{equation}
\left[\begin{array}{c}
\sigma_{11}\\
\sigma_{22}\\
\sigma_{33}\\
\sqrt{2}\sigma_{23}\\
\sqrt{2}\sigma_{13}\\
\sqrt{2}\sigma_{12}\end{array}\right]=\left[\begin{array}{cccccc}
c_{1111} & c_{1122} & c_{1133} & \sqrt{2}c_{1123} & \sqrt{2}c_{1113} & \sqrt{2}c_{1112}\\
c_{1122} & c_{2222} & c_{2233} & \sqrt{2}c_{2223} & \sqrt{2}c_{2213} & \sqrt{2}c_{2212}\\
c_{1133} & c_{2233} & c_{3333} & \sqrt{2}c_{3323} & \sqrt{2}c_{3313} & \sqrt{2}c_{3312}\\
\sqrt{2}c_{1123} & \sqrt{2}c_{2223} & \sqrt{2}c_{3323} & 2c_{2323} & 2c_{2313} & 2c_{2312}\\
\sqrt{2}c_{1113} & \sqrt{2}c_{2213} & \sqrt{2}c_{3313} & 2c_{2313} & 2c_{1313} & 2c_{1312}\\
\sqrt{2}c_{1112} & \sqrt{2}c_{2212} & \sqrt{2}c_{3312} & 2c_{2312} & 2c_{1312} & 2c_{1212}\end{array}\right]\left[\begin{array}{c}
\varepsilon_{11}\\
\varepsilon_{22}\\
\varepsilon_{33}\\
\sqrt{2}\varepsilon_{23}\\
\sqrt{2}\varepsilon_{13}\\
\sqrt{2}\varepsilon_{12}\end{array}
\right]
\,,
\label{eq:Chapman}
\end{equation}
where the elasticity tensor, whose components constitute the $6\times6$ matrix,~$C$\,, is positive-definite.
This expression is equivalent to the canonical form of Hooke's law stated in expression~(\ref{eq:Hooke}).
In expression~(\ref{eq:Chapman}), the elasticity tensor,~$c_{ijk\ell}$\,, which in its canonical form is a fourth-rank tensor in three dimensions, is expressed as a second-rank tensor in six dimensions, and equations~(\ref{eq:Chapman}) constitute  tensor equations (e.g., Chapman (2004, Section~4.4.2) and Slawinski (2015, Section~5.2.5)). This formulation is referred to as Kelvin's notation.
A common notation, known as Voigt's notation, does not constitute a tensor equation.

To apply the averaging process for a stack of generally anisotropic layers, we express equations~(\ref{eq:Chapman}) in such a manner that the left-hand sides of each equation consist of rapidly varying stresses or strains and the right-hand sides consist of algebraic combinations of rapidly varying layer-elasticity parameters multiplied by slowly varying stresses or strains.

First, consider the equations for $\sigma_{33}$, $\sigma_{23}$ and $\sigma_{13}$, which
can be written as
\begin{align*}
\sigma_{33}=&\,c_{1133}\varepsilon_{11}+c_{2233}\varepsilon_{22}+c_{3333}
\varepsilon_{33}+\sqrt{2}c_{3323}\sqrt{2}\varepsilon_{23}+\sqrt{2}c_{3313}\sqrt{2}\varepsilon_{13}
+\sqrt{2}c_{3312}\sqrt{2}\varepsilon_{12}\\[8pt]
\sqrt{2}\sigma_{23}=&\,\sqrt{2}c_{1123}\varepsilon_{11}+\sqrt{2}c_{2223}\varepsilon_{22}+\sqrt{2}c_{3323}
\varepsilon_{33}+2c_{2323}\sqrt{2}\varepsilon_{23}+2c_{2313}\sqrt{2}\varepsilon_{13}
+2c_{2312}\sqrt{2}\varepsilon_{12}\\[8pt]
\sqrt{2}\sigma_{13}=&\,\sqrt{2}c_{1113}\varepsilon_{11}+\sqrt{2}c_{2213}\varepsilon_{22}+\sqrt{2}c_{3313}
\varepsilon_{33}+2c_{2313}\sqrt{2}\varepsilon_{23}+2c_{1313}\sqrt{2}\varepsilon_{13}
+2c_{1312}\sqrt{2}\varepsilon_{12}\,,
\end{align*}
which then can be written as the matrix equation,
{\small
\begin{align}
\nonumber\underbrace{
\begin{pmatrix}
c_{3333}&\sqrt{2}c_{3323}&\sqrt{2}c_{3313}\\[8pt]
\sqrt{2}c_{3323}&2c_{2323}&2c_{2313}\\[8pt]
\sqrt{2}c_{3313}&2c_{2313}&2c_{1313}
\end{pmatrix}}_M
\underbrace{
\begin{pmatrix}
\varepsilon_{33}\\[8pt]
\sqrt{2}\varepsilon_{23}\\[8pt]
\sqrt{2}\varepsilon_{13}
\end{pmatrix}}_E
=&
\underbrace{
\begin{pmatrix}
\sigma_{33}-c_{1133}\varepsilon_{11}-c_{2233}\varepsilon_{22}-\sqrt{2}c_{3312}\sqrt{2}\varepsilon_{12}\\[8pt]
\sqrt{2}\sigma_{23}-\sqrt{2}c_{1123}\varepsilon_{11}-\sqrt{2}c_{2223}\varepsilon_{22}-2c_{2312}\sqrt{2}\varepsilon_{12}\\[8pt]
\sqrt{2}\sigma_{13}-\sqrt{2}c_{1113}\varepsilon_{11}-\sqrt{2}c_{2213}\varepsilon_{22}-2c_{1312}\sqrt{2}\varepsilon_{12}
\end{pmatrix}}_A\\
=&
\underbrace{
\begin{pmatrix}
\sigma_{33}\\[8pt]
\sqrt{2}\sigma_{23}\\[8pt]
\sqrt{2}\sigma_{13}
\end{pmatrix}}_G
-
\underbrace{
\begin{pmatrix}
c_{1133}&c_{2233}&\sqrt{2}c_{3312}\\[8pt]
\sqrt{2}c_{1123}&\sqrt{2}c_{2223}&2c_{2312}\\[8pt]
\sqrt{2}c_{1113}&\sqrt{2}c_{2213}&2c_{1312}
\end{pmatrix}}_B
\underbrace{
\begin{pmatrix}
\varepsilon_{11}\\[8pt]
\varepsilon_{22}\\[8pt]
\sqrt{2}\varepsilon_{12}
\end{pmatrix}}_F
\,.
\label{eq:MEA}
\end{align}}
$M$ is invertible, since it is positive-definite and, hence, its determinant is strictly positive.
This  positive definiteness follows from  the positive definiteness of $C$\,, given in expression~(\ref{eq:Chapman}), for $x\in{\mathbb R}^3\backslash\{0\}$ and  $y:=[0,0,x^t,0]^t$\,, $x^tMx=y^tCy>0$ as $y\neq 0.$   This follows only if $C$ is in Kelvin notation, and allows us to conclude that---since the positive definiteness is the sole constraint on the values of elasticity parameters---the Backus average is allowed for any sequence of layers composed of Hookean solids.

Notably, determinants of $M$\,, in expression~(\ref{eq:MEA}), differ by a factor of four between Voigt's notation and Kelvin's notation, used herein.
The final expressions for the equivalent medium, however, appear to be the same for both notations.

Multiplying both sides of equation~(\ref{eq:MEA}) by $M^{-1}$\,, we express the rapidly varying $E$ as
\begin{equation}
\label{eq:E}
	E=M^{-1}A=M^{-1}(G-BF)=M^{-1}G-(M^{-1}B)F\,,
\end{equation}
which means that
\[
M^{-1}G=E+(M^{-1}B)F\,,
\]
and can be averaged to get
\[
\overline{M^{-1}}\,\overline G\approx\overline E+\overline{(M^{-1}B)}\,\overline F\,,
\]
and, hence, effectively,
\begin{equation}
\overline G=\overline{(M^{-1})}^{\,\,-1}\left[\overline E+\overline{(M^{-1}B)}\,\overline{F}
\right]=\overline{(M^{-1})}^{\,\,-1}\overline E+ \overline{(M^{-1})}^{\,\,-1}
 \overline{(M^{-1}B)}\,\overline{F}
 \,.
\label{eq:G}
\end{equation}
Comparing expression~(\ref{eq:G}) with the pattern of the corresponding three lines of $C$ in expression~(\ref{eq:Chapman}), we obtain formul{\ae} for the equivalent-medium elasticity parameters.

To obtain the remaining formul{\ae}, let us examine the equations for 
the rapidly varying $\sigma_{11}$, $\sigma_{22}$ and $\sigma_{12}$\,,
which, from equation~(\ref{eq:Chapman}), can be written as
{\small
\begin{equation}
\label{eq:HJF}
\underbrace{
\begin{pmatrix}
\sigma_{11}\\[8pt]
\sigma_{22}\\[8pt]
\sqrt{2}\sigma_{12}
\end{pmatrix}}_H
=
\underbrace{
\begin{pmatrix}
c_{1111}&c_{1122}&\sqrt{2}c_{1112}\\[8pt]
c_{1122}&c_{2222}&\sqrt{2}c_{2212}\\[8pt]
\sqrt{2}c_{1112}&\sqrt{2}c_{2212}&2c_{1212}
\end{pmatrix}}_J
\underbrace{
\begin{pmatrix}
\varepsilon_{11}\\[8pt]
\varepsilon_{22}\\[8pt]
\sqrt{2}\varepsilon_{12}
\end{pmatrix}}_F
+
\underbrace{
\begin{pmatrix}
c_{1133}&\sqrt{2}c_{1123}&\sqrt{2}c_{1113}\\[8pt]
c_{2233}&\sqrt{2}c_{2223}&\sqrt{2}c_{2213}\\[8pt]
\sqrt{2}c_{3312}&2c_{2312}&2c_{1312}
\end{pmatrix}}_K
\underbrace{
\begin{pmatrix}
\varepsilon_{33}\\[8pt]
\sqrt{2}\varepsilon_{23}\\[8pt]
\sqrt{2}\varepsilon_{13}
\end{pmatrix}}_E\,.
\end{equation}}
Note that $K=B^t$\,.
Substituting expression~(\ref{eq:E}) for $E$\,, we get
\begin{equation*}
H=JF+KM^{-1}(G-BF)=JF+KM^{-1}G-KM^{-1}BF\,.
\end{equation*}
Averaging, we get
\begin{align}
\overline H\approx\,&\overline J\,\overline F+\overline{KM^{-1}}\,\overline G-
\overline{KM^{-1}B}\,\overline F\nonumber\\
=&\,(\overline J-\overline{KM^{-1}B})\overline F+\overline{KM^{-1}}\left\{
\overline{(M^{-1})}^{\,\,-1}\left[\overline E+\overline{(M^{-1}B)}\,\overline{F}\right]\right\}\nonumber\\
=&
\left[\overline J-\overline{KM^{-1}B}+\overline{KM^{-1}}\,\overline{(M^{-1})}^{\,\,-1}\overline{(M^{-1}B)}
\right]\overline F+\overline{KM^{-1}}\,\overline{(M^{-1})}^{\,\,-1}\overline E
\,.
\label{eq:H}
\end{align}
Comparing equation~(\ref{eq:H}) with  the pattern of the corresponding three lines in equation~(\ref{eq:Chapman}), we obtain formul{\ae} for the remaining equivalent-medium parameters.

We do not list in detail the formul{\ae} for the twenty-one equivalent-medium elasticity parameters of a generally anisotropic solid,
since just one such parameter takes about half-a-dozen pages.
However, a symbolic-calculation software can be used to obtain those parameters.
In Section~\ref{sec:HighSym}, we use the monoclinic symmetry to exemplify the process and list in detail the resulting formul{\ae}, and we also summarize the results for orthotropic symmetry.

The results of this section are similar to the results of Schoenberg and Muir (1989), Helbig and Schoenberg (1987, Appendix),  Helbig (1998), Carcione et al.\  (2012) and Kumar (2013), except that the tensorial form of equation~(\ref{eq:Chapman}) requires factors of $2$ and $\sqrt 2$ in several entries of $M$\,, $B$\,, $J$ and $K$\,. 
This notation allows for a convenient study of rotations, which arise in the study of elasticity tensors expressed in coordinate systems of arbitrary orientations.
%%%%%%%%%%%%%%%%%%%%%%%%%%%%
\section{Reduction to higher symmetries}
\label{sec:HighSym}
%%%%%%%%%%%%%%%%%%%%%%%%%%%%
\subsection{Monoclinic symmetry}
%%%%%%%%%%%%%%%%%%%%%%%%%%%%
\label{sec:mono}
Let us reduce the expressions derived for general anisotropy to higher material symmetries.
To do so, let us first consider the case of monoclinic layers.

The components of a monoclinic tensor can be written in a matrix form as
\begin{equation}
\label{eq:mono}
C^{\rm mono}=
\left(\begin{array}{cccccc}
c_{1111} & c_{1122} & c_{1133} & 0 & 0 &
\sqrt{2}c_{1112}\\
c_{1122} & c_{2222} & c_{2233} & 0 & 0 & \sqrt{2}c_{2212}\\
c_{1133} & c_{2233} & c_{3333} & 0 & 0 & \sqrt{2}c_{3312}\\
0 & 0 & 0 & 2c_{2323} & 2c_{2313} & 0\\
0 & 0 & 0 & 2c_{2313} & 2c_{1313} & 0\\
\sqrt{2}c_{1112} & \sqrt{2}c_{2212} & \sqrt{2}c_{3312} & 0 & 0 & 2c_{1212}
\end{array}\right)
\,;
\end{equation}
this expression corresponds to the coordinate system whose $x_3$-axis is normal to the symmetry plane.  
Inserting these components into expression~(\ref{eq:MEA}), we write
\begin{equation*}
M=
\begin{pmatrix}
c_{3333}&0&0\\[8pt]
0&2c_{2323}&2c_{2313}\\[8pt]
0&2c_{2313}&2c_{1313}
\end{pmatrix}
\,,\qquad
M^{-1}=
\begin{pmatrix}
\dfrac{1}{c_{3333}}&0&0\\[10pt]
0&\dfrac{c_{1313}}{D}&-\dfrac{c_{2313}}{D}\\[10pt]
0&-\dfrac{c_{2313}}{D}&\dfrac{c_{2323}}{D}
\end{pmatrix}\,,
\end{equation*}
where $D\equiv 2(c_{2323}c_{1313}-c_{2313}^2)$\,.
Then, we have
\begin{equation*}
\overline{M^{-1}}=
\begin{pmatrix}
\,\overline{\dfrac{1}{c_{3333}}}&0&0\\[10pt]
0&\overline{\dfrac{c_{1313}}{D}}&\overline{-\dfrac{c_{2313}}{D}}\\[10pt]
0&\overline{-\dfrac{c_{2313}}{D}}&\overline{\dfrac{c_{2323}}{D}}
\end{pmatrix}\,,
\qquad
\overline{(M^{-1})}^{\,\,-1}=
\begin{pmatrix}
\,\overline{\left(\dfrac{1}{c_{3333}}\right)}^{\,\,-1}&0&0\\[10pt]
0&\dfrac{\overline{\left(\dfrac{c_{2323}}{D}\right)}}{D_2}&\dfrac{\overline{\left(\dfrac{c_{2313}}{D}\right)}}{D_2}\\[10pt]
0&\dfrac{\overline{\left(\dfrac{c_{2313}}{D}\right)}}{D_2}&\dfrac{\overline{\left(\dfrac{c_{1313}}{D}\right)}}{D_2}
\end{pmatrix}\,,
\end{equation*}
where $D_2\equiv \overline{(c_{1313}/D)}\,\overline{(c_{2323}/D)}-\overline{(c_{2313}/D)}^2$\,.
We also have
\begin{equation*}
B=
\begin{pmatrix}
c_{1133}&c_{2233}&\sqrt{2}c_{3312}\\[8pt]
0&0&0\\[8pt]
0&0&0
\end{pmatrix}
\,,
\end{equation*}
which leads to
\begin{equation*}
M^{-1}B=
\begin{pmatrix}
\dfrac{c_{1133}}{c_{3333}}&\dfrac{c_{2233}}{c_{3333}}&\dfrac{\sqrt{2}c_{3312}}{c_{3333}}\\[16pt]
0&0&0\\[8pt]
0&0&0
\end{pmatrix}
\,,
\qquad
\overline{M^{-1}B}=
\begin{pmatrix}
\,\overline{\dfrac{c_{1133}}{c_{3333}}}&\overline{\dfrac{c_{2233}}{c_{3333}}}&
\overline{\dfrac{\sqrt{2}c_{3312}}{c_{3333}}}\,\\[16pt]
0&0&0\\[8pt]
0&0&0
\end{pmatrix}
\,.
\end{equation*}
Furthermore,
\begin{equation*}
\overline{(M^{-1})}^{\,\,-1}\overline{(M^{-1}B)}
=
\begin{pmatrix}
\,\overline{\left(\dfrac{1}{c_{3333}}\right)}^{\,\,-1}\,
\overline{\left(\dfrac{c_{1133}}{c_{3333}}\right)}&
\overline{\left(\dfrac{1}{c_{3333}}\right)}^{\,\,-1}\,
\overline{\left(\dfrac{c_{2233}}{c_{3333}}\right)}&
\overline{\left(\dfrac{1}{c_{3333}}\right)}^{\,\,-1}\,
\overline{\left(\dfrac{\sqrt{2}c_{3312}}{c_{3333}}\right)}\,
\\[16pt]
0&0&0\\[8pt]
0&0&0
\end{pmatrix}
\,.
\end{equation*}
Then, if we write equation~(\ref{eq:G}) as
\begin{equation*}
\begin{pmatrix}
\overline{\sigma_{33}}\\[6pt]
\overline{\sqrt{2}\sigma_{23}}\\[6pt]
\overline{\sqrt{2}\sigma_{13}}
\end{pmatrix}
=
\overline{(M^{-1})}^{\,\,-1}
\begin{pmatrix}
\overline{\varepsilon_{33}}\\[6pt]
\overline{\sqrt{2}\varepsilon_{23}}\\[6pt]
\overline{\sqrt{2}\varepsilon_{13}}
\end{pmatrix}
+
\overline{(M^{-1})}^{\,\,-1}\overline{(M^{-1}B)}
\begin{pmatrix}
\overline{\varepsilon_{11}}\\
\overline{\varepsilon_{22}}\\[4pt]
\overline{\sqrt{2}\varepsilon_{12}}
\end{pmatrix}
\end{equation*}
and compare it to equation~(\ref{eq:Chapman}), we obtain
\begin{equation*}
\langle c_{3333}\rangle=\overline{\left(\frac{1}{c_{3333}}\right)}^{\,\,-1}\,,
\qquad
\langle c_{2323}\rangle=\frac{\overline{\left(\dfrac{c_{2323}}{D}\right)}}{2D_2}\,,
\end{equation*} 
\begin{equation*}
\langle c_{1313}\rangle=\frac{\overline{\left(\dfrac{c_{1313}}{D}\right)}}{2D_2}\,,
\qquad
\langle c_{2313}\rangle=\frac{\overline{\left(\dfrac{c_{2313}}{D}\right)}}{2D_2}\,,
\end{equation*}
\begin{equation*}
\langle c_{1133}\rangle=
\overline{\left(\frac{1}{c_{3333}}\right)}^{\,\,-1}\,
\overline{\left(\frac{c_{1133}}{c_{3333}}\right)}\,,\!
\quad
\langle c_{2233}\rangle=
\overline{\left(\frac{1}{c_{3333}}\right)}^{\,\,-1}\,
\overline{\left(\frac{c_{2233}}{c_{3333}}\right)}\,,\!
\quad
\langle c_{3312}\rangle=
\overline{\left(\frac{1}{c_{3333}}\right)}^{\,\,-1}\,
\overline{\left(\frac{c_{3312}}{c_{3333}}\right)}\,,
\end{equation*}
where angle brackets denote the equivalent-medium elasticity parameters.

To calculate the remaining equivalent elasticity parameters from equation~(\ref{eq:H}), we insert components~(\ref{eq:mono}) into expression~(\ref{eq:HJF}) to write
\begin{equation*}
J=
\begin{pmatrix}
c_{1111}&c_{1122}&\sqrt{2}c_{1112}\\[8pt]
c_{1122}&c_{2222}&\sqrt{2}c_{2212}\\[8pt]
\sqrt{2}c_{1112}&\sqrt{2}c_{2212}&2c_{1212}
\end{pmatrix}
\,,
\quad
\overline J=
\begin{pmatrix}
\overline{c_{1111}}&\overline{c_{1122}}&\sqrt{2}\overline{c_{1112}}\\[8pt]
\overline{c_{1122}}&\overline{c_{2222}}&\sqrt{2}\overline{c_{2212}}\\[8pt]
\sqrt{2}\overline{c_{1112}}&\sqrt{2}\overline{c_{2212}}&2\overline{c_{1212}}
\end{pmatrix}
\,,
\end{equation*}
\begin{equation*}
K=
\begin{pmatrix}
c_{1133}&0&0\\[8pt]
c_{2233}&0&0\\[8pt]
\sqrt{2}c_{3312}&0&0
\end{pmatrix}
\,,
\quad
KM^{-1}=
\begin{pmatrix}
\dfrac{c_{1133}}{c_{3333}}&0&0\\[16pt]
\dfrac{c_{2233}}{c_{3333}}&0&0\\[16pt]
\sqrt{2}\,\dfrac{c_{3312}}{c_{3333}}&0&0
\end{pmatrix}
\,,
\quad
\overline{KM^{-1}}=
\begin{pmatrix}
\overline{\left(\dfrac{c_{1133}}{c_{3333}}\right)}&0&0\\[16pt]
\overline{\left(\dfrac{c_{2233}}{c_{3333}}\right)}&0&0\\[16pt]
\sqrt{2}\,\overline{\left(\dfrac{c_{3312}}{c_{3333}}\right)}&0&0
\end{pmatrix}
\,,
\end{equation*}
\begin{equation*}
KM^{-1}B=
\begin{pmatrix}
\dfrac{c_{1133}^2}{c_{3333}}&\dfrac{c_{1133}c_{2233}}{c_{3333}}&\sqrt{2}\,\dfrac{c_{3312}c_{1133}}{c_{3333}}\\[16pt]
\dfrac{c_{1133}c_{2233}}{c_{3333}}&\dfrac{c_{2233}^2}{c_{3333}}&\sqrt{2}\,\dfrac{c_{3312}c_{2233}}{c_{3333}}\\[16pt]
\sqrt{2}\,\dfrac{c_{3312}c_{1133}}{c_{3333}}&\sqrt{2}\,\dfrac{c_{3312}c_{2233}}{c_{3333}}&2\,\dfrac{c_{3312}^2}{c_{3333}}
\end{pmatrix}
\,,
\end{equation*}
\vspace{0.1in}
\begin{equation*}
\overline{KM^{-1}B}=
\begin{pmatrix}
\overline{\left(\dfrac{c_{1133}^2}{c_{3333}}\right)}&\overline{\left(\dfrac{c_{1133}c_{2233}}{c_{3333}}\right)}&\sqrt{2}\,\overline{\left(\dfrac{c_{3312}c_{1133}}{c_{3333}}\right)}\\[16pt]
\overline{\left(\dfrac{c_{1133}c_{2233}}{c_{3333}}\right)}&\overline{\left(\dfrac{c_{2233}^2}{c_{3333}}\right)}&\sqrt{2}\,\overline{\left(\dfrac{c_{3312}c_{2233}}{c_{3333}}\right)}\\[16pt]
\sqrt{2}\,\overline{\left(\dfrac{c_{3312}c_{1133}}{c_{3333}}\right)}&\sqrt{2}\,\overline{\left(\dfrac{c_{3312}c_{2233}}{c_{3333}}\right)}&2\,\overline{\left(\dfrac{c_{3312}^2}{c_{3333}}\right)}
\end{pmatrix}
\,,
\end{equation*}
%%%%
\begin{equation*}
\overline{KM^{-1}}\,\overline{(M^{-1})}^{\,\,-1}=
\begin{pmatrix}
\overline{\left(\dfrac{c_{1133}}{c_{3333}}\right)}\,\overline{\left(\dfrac{1}{c_{3333}}\right)}^{\,\,-1}&0&0\\[16pt]
\overline{\left(\dfrac{c_{2233}}{c_{3333}}\right)}\,\overline{\left(\dfrac{1}{c_{3333}}\right)}^{\,\,-1}&0&0\\[16pt]
\sqrt{2}\,\overline{\left(\dfrac{c_{3312}}{c_{3333}}\right)}\,\overline{\left(\dfrac{1}{c_{3333}}\right)}^{\,\,-1}&0&0
\end{pmatrix}
\,,
\end{equation*}
\vspace{0.1in}
\begin{align*}
&\overline{KM^{-1}}\,\overline{(M^{-1})}^{\,\,-1}\overline{M^{-1}B}\\[4pt]
&{\footnotesize\!=\begin{pmatrix}
\overline{\left(\dfrac{1}{c_{3333}}\right)}^{\,\,-1}
\overline{\left(\dfrac{c_{1133}}{c_{3333}}\right)}^2&
\overline{\left(\dfrac{1}{c_{3333}}\right)}^{\,\,-1}
\overline{\left(\dfrac{c_{1133}}{c_{3333}}\right)}\,
\overline{\left(\dfrac{c_{2233}}{c_{3333}}\right)}
&
\sqrt{2}\,
\overline{\left(\dfrac{1}{c_{3333}}\right)}^{\,\,-1}\,
\overline{\left(\dfrac{c_{1133}}{c_{3333}}\right)}\,
\overline{\left(\dfrac{c_{3312}}{c_{3333}}\right)}\,
\\[16pt]
\overline{\left(\dfrac{1}{c_{3333}}\right)}^{\,\,-1}
\overline{\left(\dfrac{c_{1133}}{c_{3333}}\right)}\,
\overline{\left(\dfrac{c_{2233}}{c_{3333}}\right)}
&
\overline{\left(\dfrac{1}{c_{3333}}\right)}^{\,\,-1}
\overline{\left(\dfrac{c_{2233}}{c_{3333}}\right)}^2&
\sqrt{2}\,
\overline{\left(\dfrac{1}{c_{3333}}\right)}^{\,\,-1}
\overline{\left(\dfrac{c_{2233}}{c_{3333}}\right)}\,
\overline{\left(\dfrac{c_{3312}}{c_{3333}}\right)}
\\[16pt]
\sqrt{2}\,
\overline{\left(\dfrac{1}{c_{3333}}\right)}^{\,\,-1}
\overline{\left(\dfrac{c_{1133}}{c_{3333}}\right)}\,
\overline{\left(\dfrac{c_{3312}}{c_{3333}}\right)}
&\sqrt{2}\,
\overline{\left(\dfrac{1}{c_{3333}}\right)}^{\,\,-1}
\overline{\left(\dfrac{c_{2233}}{c_{3333}}\right)}\,
\overline{\left(\dfrac{c_{3312}}{c_{3333}}\right)}
&
2\,
\overline{\left(\dfrac{1}{c_{3333}}\right)}^{\,\,-1}
\overline{\left(\dfrac{c_{3312}}{c_{3333}}\right)}^2
\end{pmatrix}}
\,.
\end{align*}
Then, if we write equation~(\ref{eq:H}) as
\begin{equation*}
\begin{pmatrix}
\overline{\sigma_{11}}\\[8pt]
\overline{\sigma_{22}}\\[8pt]
\sqrt{2}\overline{\sigma_{12}}
\end{pmatrix}
=
\left[\overline J-\overline{KM^{-1}B}+\overline{KM^{-1}}\,\overline{(M^{-1})}^{\,\,-1}\overline{(M^{-1}B)}
\right]\!
\begin{pmatrix}
\overline{\varepsilon_{11}}\\[8pt]
\overline{\varepsilon_{22}}\\[8pt]
\sqrt{2}\overline{\varepsilon_{12}}
\end{pmatrix}
+
\overline{KM^{-1}}\,\overline{(M^{-1})}^{\,\,-1}\!
\begin{pmatrix}
\overline{\varepsilon_{33}}\\[8pt]
\sqrt{2}\overline{\varepsilon_{23}}\\[8pt]
\sqrt{2}\overline{\varepsilon_{13}}
\end{pmatrix}
\end{equation*}
and compare it to equation~(\ref{eq:Chapman}), we obtain
\begin{equation*}
\langle c_{1133}\rangle=
\overline{\left(\frac{1}{c_{3333}}\right)}^{\,\,-1}
\overline{\left(\frac{c_{1133}}{c_{3333}}\right)}\,,\!
\quad
\langle c_{2233}\rangle=
\overline{\left(\frac{1}{c_{3333}}\right)}^{\,\,-1}
\overline{\left(\frac{c_{2233}}{c_{3333}}\right)}\,,\!
\quad
\langle c_{3312}\rangle=
\overline{\left(\frac{1}{c_{3333}}\right)}^{\,\,-1}
\overline{\left(\frac{c_{3312}}{c_{3333}}\right)}\,,
\end{equation*}
as before, and
\begin{equation*}
\langle c_{1111}\rangle=
\overline{c_{1111}}-\overline{\left(\frac{c_{1133}^2}{c_{3333}}\right)}+
\overline{\left(\frac{1}{c_{3333}}\right)}^{\,\,-1}
\overline{\left(\frac{c_{1133}}{c_{3333}}\right)}^{\,2}\,,
\end{equation*}
\begin{equation*}
\langle c_{1122}\rangle=
\overline{c_{1122}}-\overline{\left(\frac{c_{1133}\,c_{2233}}{c_{3333}}\right)}+
\overline{\left(\frac{1}{c_{3333}}\right)}^{\,\,-1}
\overline{\left(\frac{c_{1133}}{c_{3333}}\right)}\,\,
\overline{\left(\frac{c_{2233}}{c_{3333}}\right)}\,,
\end{equation*}
\begin{equation*}
\langle c_{2222}\rangle=
\overline{c_{2222}}-\overline{\left(\frac{c_{2233}^2}{c_{3333}}\right)}+
\overline{\left(\frac{1}{c_{3333}}\right)}^{\,\,-1}
\overline{\left(\frac{c_{2233}}{c_{3333}}\right)}^{\,2}\,,
\end{equation*}
\begin{equation*}
\langle c_{1212}\rangle=
\overline{c_{1212}}-\overline{\left(\frac{c_{3312}^2}{c_{3333}}\right)}+
\overline{\left(\frac{1}{c_{3333}}\right)}^{\,\,-1}
\overline{\left(\frac{c_{3312}}{c_{3333}}\right)}^{\,2}\,,
\end{equation*}
\begin{equation*}
\langle c_{1112}\rangle=
\overline{c_{1112}}-\overline{\left(\frac{c_{3312}\,c_{1133}}{c_{3333}}\right)}+
\overline{\left(\frac{1}{c_{3333}}\right)}^{\,\,-1}
\overline{\left(\frac{c_{1133}}{c_{3333}}\right)}\,\,
\overline{\left(\frac{c_{3312}}{c_{3333}}\right)}\,,
\end{equation*}
\begin{equation*}
\langle c_{2212}\rangle=
\overline{c_{2212}}-\overline{\left(\frac{c_{3312}\,c_{2233}}{c_{3333}}\right)}+
\overline{\left(\frac{1}{c_{3333}}\right)}^{\,\,-1}
\overline{\left(\frac{c_{2233}}{c_{3333}}\right)}\,\,
\overline{\left(\frac{c_{3312}}{c_{3333}}\right)}\,.
\end{equation*}
The other equivalent-medium elasticity parameters are zero. 
Thus, we have thirteen linearly independent parameters in the form of matrix~(\ref{eq:mono}).
Hence, the equivalent medium exhibits the same symmetry as the individual layers.
Also if we set $c_{1112}$, $c_{2212}$, $c_{3312}$ and $c_{2313}$ to zero the
results of this section reduce to the results of the next section.    The results of
this section differ from the results of Kumar (2013, Appendix~B) but that is
because that paper uses a vertical ($x_1$-$x_3$) symmetry plane whereas
we use a horizontal ($x_1$-$x_2$) symmetry plane, which---since it is parallel
to the layering---produces simpler results.
%%%%%%%%%%%%%%%%%%%%%%%%%%%%
\subsection{Orthotropic symmetry}
%%%%%%%%%%%%%%%%%%%%%%%%%%%%
Continuing the reduction of expressions derived for general anisotropy to higher material symmetries, let us consider the case of orthotropic layers.
The components of an orthotropic tensor can be written as
\begin{equation}
C^{\rm ortho}=
\label{eq:ortho}
\begin{pmatrix}
c_{1111}&c_{1122}&c_{1133}&0&0&0\\
c_{1122}&c_{2222}&c_{2233}&0&0&0\\
c_{1133}&c_{2233}&c_{3333}&0&0&0\\
0&0&0&2c_{2323}&0&0\\
0&0&0&0&2c_{1313}&0\\
0&0&0&0&0&2c_{1212}
\end{pmatrix}\,;
\end{equation}
this equation corresponds to the coordinate system whose axes are
normal to the symmetry planes.

The equivalent medium elasticity parameters can be derived in a similar manner
as in section~\ref{sec:mono} or by setting $c_{1112}$, $c_{2212}$, $c_{3312}$ and $c_{2313}$ to zero in the results of section~\ref{sec:mono} .  In either case we obtain
\begin{equation*}
\langle c_{3333}\rangle=\overline{\left(\frac{1}{c_{3333}}\right)}^{\,\,-1}\,,
\qquad
\langle c_{2323}\rangle=\overline{\left(\frac{1}{c_{2323}}\right)}^{\,\,-1}\,,
\qquad
\langle c_{1313}\rangle=\overline{\left(\frac{1}{c_{1313}}\right)}^{\,\,-1}\,,
\end{equation*}
\begin{equation*}
\langle c_{1133}\rangle=
\overline{\left(\frac{1}{c_{3333}}\right)}^{\,\,-1}
\overline{\left(\frac{c_{1133}}{c_{3333}}\right)}\,,
\qquad
\langle c_{2233}\rangle=
\overline{\left(\frac{1}{c_{3333}}\right)}^{\,\,-1}
\overline{\left(\frac{c_{2233}}{c_{3333}}\right)}\,,
\end{equation*}
\begin{equation*}
\langle c_{1111}\rangle=
\overline{c_{1111}}-\overline{\left(\frac{c_{1133}^2}{c_{3333}}\right)}+
\overline{\left(\frac{1}{c_{3333}}\right)}^{\,\,-1}
\overline{\left(\frac{c_{1133}}{c_{3333}}\right)}^2\,,
\end{equation*}
\begin{equation*}
\langle c_{1122}\rangle=
\overline{c_{1122}}-\overline{\left(\frac{c_{1133}\,c_{2233}}{c_{3333}}\right)}+
\overline{\left(\frac{1}{c_{3333}}\right)}^{\,\,-1}
\overline{\left(\frac{c_{1133}}{c_{3333}}\right)}\,
\overline{\left(\frac{c_{2233}}{c_{3333}}\right)}\,,
\end{equation*}
\begin{equation*}
\langle c_{2222}\rangle=
\overline{c_{2222}}-\overline{\left(\frac{c_{2233}^2}{c_{3333}}\right)}+
\overline{\left(\frac{1}{c_{3333}}\right)}^{\,\,-1}
\overline{\left(\frac{c_{2233}}{c_{3333}}\right)}^2\,,\qquad
\langle c_{1212}\rangle=\overline{c_{1212}}\,.
\end{equation*}
The other equivalent-medium elasticity parameters are zero. 
Thus, we have nine linearly independent parameters in the form of matrix~(\ref{eq:ortho}).
Hence, the equivalent medium exhibits the same symmetry as the individual layers.
Subsequent reductions to transversely isotropic and isotropic layers result, respectively, in expressions~(9) and (13) of Backus (1962).

Also, the results of this section agree with the results of Tiwary (2007, expression~(5.1)) except
for the fifth equation of that expression, which contains a typo:  $C_{13}$ instead of a $C_{23}$\,. 
Tiwary (2007) references that expression to Shermergor (1977, expression~(2.4)), a book in Russian; since we do not have access to that book, we cannot ascertain whether or not that typo originates with Shermergor (1977).   The results of this section also agree with the
results of Kumar (2013, Appendix~B) and of Slawinski (2016, Exercise~4.6).
%%%%%
%%%%%%%%%%%%%%%%%%%%%%%%%%%%
\section{Conclusions}
%%%%%%%%%%%%%%%%%%%%%%%%%%%%
In this paper, using the case of the medium that is a long-wave equivalent of a stack of thin generally anisotropic layers, we examine the mathematical underpinnings of the approach of Backus (1962), whose underlying assumption remains lateral homogeneity.

Following explicit statements of assumptions and definitions, in Lemma~\ref{lem:LemStab}, we prove---within the long-wave approximation---that if the thin layers obey stability conditions then so does the equivalent medium.
Also, we show that the Backus average is allowed for any sequence of layers composed of Hookean solids.
As a part of the discussion of approximations, in the proof of Lemma~\ref{lem:LemProd}, we examine---within the Backus-average context---the approximation of the average of a product as the product of averages, and give upper bounds for their difference in Propositions~\ref{prop:One} and \ref{prop:Two}.
%%%%%%%%%%%%%%%%%%%%%%%%%%%%
\section{Further work}
%%%%%%%%%%%%%%%%%%%%%%%%%%%%
The subject of Backus average was examined by several researchers, among them,  Helbig and Schoenberg (1987), Schoenberg and Muir (1989), Berryman (1997), Helbig (1998, 2000), Carcione et al.\  (2012), Kumar (2013), Brisco (2014), and Danek and Slawinski (2016).
However, further venues of investigation remain open.

A following step is the error-propagation analysis, which is the effect of errors in layer parameters on the errors of the equivalent medium.
This step might be performed with perturbation techniques.
Also, using such techniques, we could examine numerically the precise validity of $\overline{fg}\approx \overline f \,\overline g$\,, which is the approximation of Lemma~\ref{lem:LemProd}.

Another numerical study could examine whether the equivalent medium for a stack of strongly anisotropic layers, whose anisotropic properties are randomly different from each other, is weakly anisotropic.
If so, we might seek---using the method proposed by Gazis et al.\ (1963) and elaborated by Danek et al.\ (2015)---an elasticity tensor of a higher symmetry that is nearest to that medium.  For
such a study, Kelvin's notation---used in this paper---is preferable, even though one could accommodate 
rotations in Voigt's notation by using the Bond (1943) transformation (e.g.,
Slawinski (2015), section~5.2).

A further possibility is an empirical examination of the obtained formul{\ae}.
This could be achieved with seismic data, where the layer properties are obtained from well-logging tools and the equivalent parameters from vertical seismic profiling.
%%%%%%%%%%%%%%%%%%%%%%%%%%%%
\section*{Acknowledgments}
%%%%%%%%%%%%%%%%%%%%%%%%%%%%
We wish to acknowledge discussions with George Backus, Klaus Helbig, Mikhail Kochetov and Michael
Rochester. This research was performed in the context of The Geomechanics Project
supported by Husky Energy. Also, this research was partially supported by the
Natural Sciences and Engineering Research Council of Canada, grant 238416-2013.
%%%%%%%%%%%%%%%%%%%%%%%%%%%%
\section*{References}
%%%%%%%%%%%%%%%%%%%%%%%%%%%%
\frenchspacing
\newcommand{\hd}{\par\noindent\hangindent=0.4in\hangafter=1}
\hd
Anderson, D.L., Elastic wave propagation in layered anisotropic media,
{\it J. Geophys. Research\/} {\bf 66}, 2953--2964, 1961.
\setlength{\parskip}{4pt}
\hd
Backus, G.E.,  Long-wave elastic anisotropy produced by horizontal layering,
{\it  J. Geophys. Res.\/} {\bf 67}, 11, 4427--4440, 1962.
\hd
Berryman, J.G., Range of the {$P$}-wave anisotropy parameter for finely layered {VTI} media,
{\it Stanford Exploration Project\/} {\bf 93}, 179--192, 1997.
\hd
Bond, W.L., The mathematics of the physical properties of crystals,
{\it Bell System Technical Journal\/} {\bf 22}, 1--72, 1943.
\hd
Brisco, C., {\it Anisotropy vs. inhomogeneity: Algorithm formulation, coding and modelling,
Honours Thesis\/}, Memorial University, 2014.
\hd
Carcione, J.M., S. Picotti, F. Cavallini and J.E. Santos, Numerical test of the
Schoenberg-Muir theory, {\it Geophysics\/} {\bf 77}, 2, C27--C35, 2012.
\hd
Chapman, C.H., {\it Waves and rays in elastic continua\/}, Cambridge University Press, 2004.
\hd
Dalton, D., and M.A. Slawinski, On Backus average for oblique incidence, {\it
ar{X}iv\/}:1601.02966v1 [physics.geo-ph], 2016.
\hd
Danek, T.,  M. Kochetov and M.A. Slawinski, Effective elasticity tensors in the context of random errors, {\it Journal of Elasticity\/} {\bf 121}(1), 55--67, 2015.
\hd
Danek, T., and M.A. Slawinski,  Backus average under random perturbations of
layered media, {\it SIAM Journal on Applied Mathematics\/}, MS\#M104317, 2016.
\hd
Gazis, D.C., I. Tadjbakhsh and R.A. Toupin, The elastic tensor of given symmetry nearest to an anisotropic elastic tensor, {\it Acta Crystallographica\/} {\bf 16}, 9, 917--922, 1963.
\hd
Haskell, N.A., Dispersion of surface waves on multilayered media, {\it Bull. Seism.
Soc. Am.\/} {\bf 43}, 17--34, 1953.
\hd
Helbig, K., Elastischen Wellen in anisotropen Medien, {\it Getlands Beitr. Geophys.\/}
{\bf 67}, 256--288, 1958.
\hd
Helbig, K., Layer-induced anisotropy: Forward relations between between constituent parameters and compound parameters,
{\it Revista Brasileira de Geof{\'\i}sica\/} {\bf 16}, 2--3, 103--114, 1998.
\hd
Helbig, K., Inversion of compound parameters to constituent parameters,
{\it Revista Brasileira de Geof{\'\i}sica\/} {\bf 18}, 2, 173--185, 2000.
\hd
Helbig, K. and M. Schoenberg, Anomalous polarization of elastic waves in transversely
isotropic media, {\it J. Acoust. Soc. Am.\/} {\bf 81}, 5, 1235--1245, 1987.
\hd
Kumar, D., Applying Backus averaging for deriving seismic anisotropy of a long-wavelength equivalent medium from well-log data, {\it  J. Geophys. Eng.\/} {\bf 10}, 1--15, 2013.
\hd
Postma, G.W., Wave propagation in a stratified medium, {\it Geophysics\/} {\bf 20},
780--806, 1955.
\hd
Riznichenko, Yu. Y., On seismic anisotropy, {\it Invest. Akad. Nauk SSSR, Ser. Geograf. i Geofiz.\/}
{\bf 13}, 518--544, 1949.
\hd
Rudzki, M.P., Parametrische {D}arstellung der elastischen {W}ellen in anisotropischen {M}edie,
{\it Bull. Acad. Cracovie\/}, 503, 1911.
\hd
Rytov, S.M., The acoustical properties of a finely layered medium, {\it Akust. Zhur.\/},
{\bf 2}, 71, 1956.   See also {\it Sov. Phys. Acoust.\/} {\bf 2}, 67, 1956.
\hd
Schoenberg, M. and F. Muir, A calculus for finely layered anisotropic media, {\it Geophysics\/}
{\bf 54}, 5, 581--589, 1989.
\hd
Shermergor, T., {\it Theory of elasticity of microinhomogeneous media\/}. (in Russian), Nauka, 1977.
\hd
Slawinski, M.A. {\it Wavefronts and rays in seismology: Answers to unasked questions\/},
World Scientific, 2016.
\hd
Slawinski, M.A., {\it Waves and rays in elastic continua\/}, World Scientific, 2015.
\hd
Tiwary, D.K., {\it Mathematical modelling and ultrasonic measurement of shale anisotropy and a 
comparison of upscaling methods from sonic to seismic\/},  Ph.D. thesis, University of Oklahoma, 2007.
\hd
Thomson, W.T., Transmission of elastic waves through a stratified solid medium, {\it J.
Appl. Phys.\/} {\bf 21}, 80--93, 1950.
\hd
White, J.E., and F.A. Angona, Elastic wave velocities in laminated media,
{\it J. Acoust. Soc. Am.\/} {\bf 27}, 310--317, 1955.
%%%%%%%%%%%%%%%%%%%%%%%%%%%%%%%%
\setcounter{section}{0}
\setlength{\parskip}{0pt}
\renewcommand{\thesection}{Appendix~\Alph{section}}
%%%%%%%%%%%%%%%%%%%%%%%%%%%%%%%%
\section{Average of derivatives}
\label{AppLemDer}
%%%%%%%%%%%%%%%%%%%%%%%%%%%%%%%%
\begin{proof}
We begin with the definition of averaging,
\begin{equation}
\label{eq:LenTrivial}
\overline{f}(x_3):=\int\limits_{-\infty}^\infty w(\xi-x_3)f(\xi)\,{\rm d}\xi
\,.
\end{equation}
The derivatives with respect to $x_1$ and $x_2$  can be written as
\begin{align*}
  \frac{\partial\overline{f}}{\partial x_i}
  &=\frac{\partial}{\partial x_i}\int\limits_{-\infty}^\infty w(\xi-x_3)f(x_1,x_2,\xi)\,{\rm d}\xi\\
  &=\int\limits_{-\infty}^\infty w(\xi-x_3)\frac{\partial f(x_1,x_2,\xi)}{\partial x_i}\,{\rm d}\xi
  =:\overline{\frac{\partial f}{\partial x_i}}
  \,,
 \qquad i=1,2
 \,,
\end{align*}
where the last equality is the statement of definition~(\ref{eq:LenTrivial}), as required.
For the derivatives with respect to $x_3$\,, we need to verify that
\begin{equation}
\label{eq:MishaNotation}
\frac{\partial}{\partial x_3}\int\limits_{-\infty}^\infty w(\xi-x_3)f(x_1,x_2,\xi)\,{\rm d}\xi
 =
 \int\limits_{-\infty}^\infty w(\xi-x_3)\frac{\partial f(x_1,x_2,\xi)}{\partial\xi}\,{\rm d}\xi
  \,.
\end{equation}
Applying integration by parts, we write the right-hand side as
\begin{equation*}
\left.w(\xi-x_3)f(x_1,x_2,\xi)\right|_{-\infty}^\infty-\int\limits_{-\infty}^\infty w'(\xi-x_3)\,f(x_1,x_2,\xi)\,{\rm d}\xi
 \,,
\end{equation*}
where $w$ is a function of a single variable.
Since
\begin{equation*}
\lim_{x_3\rightarrow\pm\infty}w(x_3)=0
\,,
\end{equation*}
the product of $w$ and $f$ vanishes at $\pm\infty$\,, and we are left with
\begin{equation*}
-\int\limits_{-\infty}^\infty w'(\xi-x_3)\,f(x_1,x_2,\xi)\,{\rm d}\xi
 \,.
\end{equation*}
Let us consider the left-hand side of expression~(\ref{eq:MishaNotation}).
Since only $w$ is a function of $x_3$\,, we can interchange the operations of integration and differentiation to write
\begin{equation*}
-\int\limits_{-\infty}^\infty w'(\xi-x_3)\,f(x_1,x_2,\xi)\,{\rm d}\xi
 \,;
\end{equation*}
the negative sign arises from the chain rule,
\begin{equation*}
\frac{\partial w(\xi -x_3)}{\partial x_3}
=w'(\xi -x_3)\frac{\partial(\xi -x_3)}{\partial x_3}
=-w'(\xi-x_3)
\,.
\end{equation*}
Thus, both sides of expression~(\ref{eq:MishaNotation}) are equal to one another, as required.
In other words,
\begin{equation*}
\frac{\partial\,\overline f}{\partial x_3}
=
 \overline{\frac{\partial f}{\partial x_3}}
 \,,
\end{equation*}
which completes the proof.
\end{proof}
%%%%%%%%%%%%%%%%%%%%%%%%%%%%
\section{Stability of equivalent medium}
\label{AppLemStab}
%%%%%%%%%%%%%%%%%%%%%%%%%%%%
\begin{proof}
The stability of layers means that their deformation requires work.
Mathematically, it means that, for each layer,
\begin{equation*}
 W=\frac{1}{2}
 \sigma\cdot\varepsilon
 >0
\,,
\end{equation*}
where $W$ stands for work, and $\sigma$ and $\varepsilon$ denote the stress and strain tensors, respectively, which are expressed as columns in equation~(\ref{eq:Chapman}): $\sigma=C\varepsilon$\,.
As an aside, we can say that, herein, $W>0$ is equivalent to the positive definiteness of $C$\,, for each layer.

Performing the average of $W$ over all layers and using---in the scalar product---the fact that the average of a sum is the sum of averages, we write
\begin{equation*}
\overline{W}=\frac{1}{2}
 \overline{\sigma\cdot\varepsilon}
 >0
\,.
\end{equation*}
Thus, $W>0\implies\overline{W}>0$\,.

Let us proceed to show that this implication---in turn---entails the stability of the equivalent medium, which is tantamount to the positive definiteness of $\langle C\,\rangle$\,.

Following Lemma~\ref{lem:LemProd}---if one of two functions is nearly constant---we can approximate the average of their product by the product of their averages,
\begin{equation}
\label{eq:MishaNov}
\overline{W}=\frac{1}{2}
 \overline{\sigma}\cdot\overline{\varepsilon}
 >0
\,.
\end{equation}
Herein, we use the property stated in Section~\ref{sub:Approx} that $\sigma_{i3}$\,, where $i\in\{1,2,3\}$\,, are constant, and $\varepsilon_{11}$\,, $\varepsilon_{12}$ and $\varepsilon_{22}$ vary slowly, along the $x_3$-axis,
together with Lemma~\ref{lem:LemProd}, which can be invoked due to the fact that each product in expression~(\ref{eq:MishaNov}) is such that one function is nearly constant and the other possibly varies more rapidly.

By definition of Hooke's law, $\overline{\sigma}:=\langle C\,\rangle\,\overline{\varepsilon}$\,, expression~(\ref{eq:MishaNov}) can be written as
\begin{equation*}
\frac{1}{2}\,
 \left(\,\langle C\,\rangle\,\overline{\varepsilon}\,\right)\cdot\overline{\varepsilon}
 >0
\,,\qquad\forall\,\,\overline{\varepsilon}\neq 0
\,,
\end{equation*}
which means that $\langle C\,\rangle$ is positive-definite, and which---in view of this derivation---proves that the equivalent medium inherits the stability of individual layers.
\end{proof}
%%%%%%%%%%%%%%%%%%%%%%%%%%%%
\section{Approximation of product}
\label{AppLemApp}
%%%%%%%%%%%%%%%%%%%%%%%%%%%%
For a fixed $x_3$\,, we may set $W(\zeta):=w(\zeta-x_3)$\,.
Then, $W\geqslant 0$ and
$\int_{-\infty}^\infty W(\zeta)\,{\rm d}\zeta =1$\,.
With this notation, equation~(\ref{eq:BackusOne}) becomes
\[
\overline f:=\int\limits_{-\infty}^\infty f(x)\,W(x)\,{\rm d}x
\,.
\]
Similarly,
\begin{equation*}
\overline g:=\int\limits_{-\infty}^\infty g(x)\,W(x)\,{\rm d}x
\qquad
{\rm and}
\qquad
\overline{fg}:=\int\limits_{-\infty}^\infty f(x)\,g(x)\,W(x)\,{\rm d}x\,.
\end{equation*}
\begin{prop}
\label{prop:One}
Suppose that the first derivatives of $f$ and $g$ are uniformly bounded; that is, both
\begin{equation*}
\|f'\|_{\infty}:=\sup_{-\infty<x<\infty}|f'(x)| 
\qquad
and
\qquad
\|g'\|_{\infty}:=\sup_{-\infty<x<\infty}|g'(x)|
\end{equation*}
are finite.
Then, we have 
\[|{\overline{fg}}-{\overline f}{\overline g}|\leqslant 2\,(\ell')^2\,\|f'\|_\infty\|g'\|_\infty\,.\]
\end{prop}
\begin{proof}
We may calculate
\begin{align*}
\int\limits_{-\infty}^\infty &\left(f(x)-\overline f\right)\left(g(x)-\overline g\right)W(x)\,{\rm d}x\\
=& \int\limits_{-\infty}^\infty f(x)\,g(x)\,W(x)\,{\rm d}x-\overline f\int\limits_{-\infty}^\infty g(x)\,W(x)\,{\rm d}x
-\overline g\int\limits_{-\infty}^\infty f(x)\,W(x)\,{\rm d}x+\int\limits_{-\infty}^\infty \overline f\,\overline g\,W(x)\,{\rm d}x\\
=& \,\overline{fg}-\overline{f}\int\limits_{-\infty}^\infty g(x)\,W(x)\,{\rm d}x
-\overline g\int\limits_{-\infty}^\infty f(x)\,W(x)\,{\rm d}x+\overline f\,\overline g
\\=&\,\overline{fg}-\overline f\, \overline g-\overline g\,\overline f+\overline f \,\overline g=
\overline{fg}-\overline f\, \overline g
\,;
\end{align*}
that is,
\begin{equation}
\overline{fg}-\overline f\, \overline g
=
\int\limits_{-\infty}^\infty\left(f(x)-\overline f\right)\left(g(x)-\overline g\right)W(x)\,{\rm d}x\,.
\label{eq:Len}
\end{equation}
Now,
\[
f(x)-\overline f=f(x)-\int\limits_{-\infty}^\infty f(y)W(y)\,{\rm d}y
=\int\limits_{-\infty}^\infty\left(f(x)-f(y)\right)W(y)\,{\rm d}y
\,,
\]
so that
\begin{equation}\label{eq:Len2}
|f(x)-\overline f|\leqslant \|f'\|_\infty\int\limits_{-\infty}^\infty |x-y|\,W(y)\,{\rm d}y
\,,
\end{equation}
and hence, by the Cauchy-Schwartz inequality,
\[|f(x)-\overline f|^2\leqslant  \|f'\|_\infty^2 \int\limits_{-\infty}^\infty |x-y|^2\,W(y)\,{\rm d}y
\, \int\limits_{-\infty}^\infty 1^2\,W(y)\,{\rm d}y= \|f'\|_\infty^2 \int\limits_{-\infty}^\infty |x-y|^2\,W(y)\,{\rm d}y\,.\]
Thus,
\begin{align*}
 \int\limits_{-\infty}^\infty |f(x)-\overline f|^2\,W(x)\,{\rm d}x&\leqslant 
 \|f'\|_\infty^2 \int\limits_{-\infty}^\infty \int\limits_{-\infty}^\infty |x-y|^2\,W(x)\,W(y)\,{\rm d}x\,{\rm d}y\\
  &= \|f'\|_\infty^2 \int\limits_{-\infty}^\infty \int\limits_{-\infty}^\infty (x^2-2xy+y^2)\,W(x)\,W(y)\,{\rm d}x\,{\rm d}y\\
  &= \|f'\|_\infty^2 \left(2 \int\limits_{-\infty}^\infty x^2\,W(x)\,{\rm d}x-2\left( \int\limits_{-\infty}^\infty
  x\,W(x)\,{\rm d}x\right)^{\!\!\!2\,\,}\right).
\end{align*}
It follows, by the Cauchy-Schwartz inequality applied to equation~\eqref{eq:Len}, that
\begin{align*}
|\overline{fg}-\overline f\, \overline g|^2&\leqslant 
 \int\limits_{-\infty}^\infty |f(x)-\overline f|^2\,W(x)\,{\rm d}x \,
  \int\limits_{-\infty}^\infty |g(x)-\overline g|^2\,W(x)\,{\rm d}x\\
  &\leqslant  \|f'\|_\infty^2  \|g'\|_\infty^2 \left(2 \int\limits_{-\infty}^\infty x^2\,W(x)\,{\rm d}x-2\left( \int\limits_{-\infty}^\infty
  x\,W(x)\,{\rm d}x\right)^{\!\!\!2\,\,}\right)^{\!\!\!2}
  \,.
\end{align*}
Note that
\[ \int\limits_{-\infty}^\infty x\,W(x)dx= \int\limits_{-\infty}^\infty x\,w(x-x_3)\,{\rm d}x
= \int\limits_{-\infty}^\infty (x+x_3)\,w(x)\,{\rm d}x=x_3\,,\]
using the defining properties of $w(\zeta).$ Similarly
\[ \int\limits_{-\infty}^\infty x^2\,W(x)\,{\rm d}x= \int\limits_{-\infty}^\infty x^2\,w(x-x_3)\,{\rm d}x=\int\limits_{-\infty}^\infty (x+x_3)^2\,w(x)\,{\rm d}x=(\ell')^2+x_3^2\,.\]
Consequently,
\[ 2 \int\limits_{-\infty}^\infty x^2\,W(x)\,{\rm d}x-2\left( \int\limits_{-\infty}^\infty
  x\,W(x)\,{\rm d}x\right)^{\!\!\!2}=2((\ell')^2+x_3^2-x_3^2)=2(\ell')^2\]
 and we have
 \[|\overline{fg}-\overline f\, \overline g|\le  2\,(\ell')^2\,\|f'\|_\infty\|g'\|_\infty\,,\]
 as claimed.
\end{proof}
\noindent Hence, if $f$ and $g$ are nearly constant, which means that $\|f'\|_\infty$ and $\|g'\|_\infty$
are small, then $\overline{fg}\approx \overline f \overline g$\,.
\begin{corollary}
Since the error estimate involves the product of the norms of the derivatives, it follows that if
one of them is small enough and the other is not excessively large, then their product can
be small enough for the approximation, $\overline{fg}\approx \overline f \,\overline g$\,, to hold.	
\end{corollary}
\noindent The exact accuracy of this property will be examined further by numerical methods in a future publication.

If $g(x)\geqslant 0$\,, we can say more, even if $g(x)$ is wildly varying.
If $f$ is continuous and $g(x)\geqslant 0$\,, then, by the Mean-value Theorem for Integrals,
\[
\overline{fg}=\int\limits_{-\infty}^\infty f(x)\,g(x)\,W(x)\,{\rm d}x=f(c)
\int\limits_{-\infty}^\infty g(x)\,W(x)\,{\rm d}x=f(c)\,\overline g
\,,
\]
for some $c$\,.  Hence,
\[
\overline{fg}-\overline f\,\overline g=f(c)\,\overline g-\overline f\,\overline g=
(f(c)-\overline f)\,\overline g\,.
\]
This implies that
\[
|\overline{fg}-\overline f\,\overline g|\leqslant|f(c)-\overline f|\,\overline g
\leqslant \|f'\|_\infty\left(\,\int\limits_{-\infty}^\infty |x-y|\,W(y)\,{\rm d}y\right)\,\overline g
\,.
\]
Hence, even for $g$ wildly varying---as long as $\overline g$ is not too big in relation to $\|f'\|_\infty$\,---it is still the case that the average of the product
is close to the product of the averages.
A bound on ${\int_{-\infty}^\infty |x-y|\,W(y)\,{\rm d}y}$ would depend on the weight function,~$w$\,, used.

An alternative estimate is provided by the following proposition.
\begin{prop}
\label{prop:Two}
Suppose that $m:=\inf f(x)>-\infty$ and $M:=\sup\limits_{-\infty<x<\infty}f(x)<\infty$ and\\${\sup\limits_{-\infty<x<\infty}|g(x)|<\infty}$\,.
Then,
\begin{equation*}
|\overline{fg}-\overline f\,\overline g|\leqslant\left(\sup\limits_{-\infty<x<\infty}|g(x)|\right)(M-m)\,.
\end{equation*}
\end{prop}
\begin{proof}
\begin{equation*}
\overline{fg}=\int\limits_{-\infty}^{\infty}f(x)\,g(x)\,W(x)\,{\rm d}x
=\int\limits_{-\infty}^{\infty}(f(x)-m)\,g(x)\,W(x)\,{\rm d}x+m\int\limits_{-\infty}^{\infty}g(x)\,W(x)\,{\rm d}x\,,
\end{equation*}
which---by the definition of the average---is
\begin{equation*}
\overline{fg}
=\int\limits_{-\infty}^{\infty}(f(x)-m)\,g(x)\,W(x)\,{\rm d}x+m\,\overline g
\,.
\end{equation*}
Hence,
\begin{align}
\label{Cassinelle1}
\nonumber|\overline{fg}-\overline f\,\overline g|
=&\left|\,\int\limits_{-\infty}^{\infty}(f(x)-m)\,g(x)\,W(x)\,{\rm d}x+\left(m-\overline f\,\right)\,\overline g\,\right|\\
\nonumber\leqslant &\left(\sup\limits_{-\infty<x<\infty}|g(x)|\right)\int\limits_{-\infty}^{\infty}(f(x)-m)\,W(x)\,{\rm d}x+\left|\,m-\overline f\,\right|\,|\overline g|\\
\nonumber=&\left(\sup\limits_{-\infty<x<\infty}|g(x)|\right)\left(\overline f-m\right)+\left(\overline f-m\right)\,|\overline g|\\
\leqslant &\,2\left(\sup\limits_{-\infty<x<\infty}|g(x)|\right)\left(\overline f-m\right)
\,.
\end{align}
Similarly,
\begin{equation}
\label{Cassinelle2}
|\overline{fg}-\overline f\,\overline g|\leqslant2\left(\sup\limits_{-\infty<x<\infty}|g(x)|\right)\left(M-\overline f\right)
\,.
\end{equation}
Taking the average of expressions~(\ref{Cassinelle1}) and (\ref{Cassinelle2}), we obtain
\begin{align*}
|\overline{fg}-\overline f\,\overline g|\leqslant &\,2\left(\sup\limits_{-\infty<x<\infty}|g(x)|\right)\frac{\left(\overline f-m\right)+\left(M-\overline f\right)}{2}\\
=&\left(\sup\limits_{-\infty<x<\infty}|g(x)|\right)(M-m)
\,,
\end{align*}
as required.
\end{proof}
\noindent Consequently, if $f(x)$ is almost constant, which means that $m\approx M$\,, then $\overline{fg}\approx\overline f\,\overline g$\,.
\end{document}